\documentclass[10pt,final,journal]{IEEEtran}
\usepackage{amsmath,latexsym,amssymb,bbm,amsthm, mathtools, verbatim, array}
\usepackage{color}

\newtheorem{theorem}{Theorem}
\newtheorem{lemma}[theorem]{Lemma}

\newtheorem{proposition}[theorem]{Proposition}
\newtheorem{definition}[theorem]{Definition}

\newlength{\blank}
\newcommand{\nc}{\newcommand}
\newcommand{\tr}{\text{tr}}
\nc{\bra}[1]{\langle#1|}
\nc{\ket}[1]{|#1\rangle}

\nc{\pw}{\mt{PW}}
\nc{\arbclass}{\mt{\Omega}}
\nc{\rnc}{\renewcommand}

\nc{\catchset}{T}

\nc{\dg}{\dagger}
\nc{\dn}{\downarrow}

\nc{\1}{{\mathbbm{1}}}
\nc{\id}{{\operatorname{id}}}

\begin{document}

\title{Public Quantum Communication and Superactivation}

\author{Fernando~G.S.L.~Brand\~ao and Jonathan Oppenheim 
  \thanks{Fernando Brand\~ao (fgslbrandao@googlemail.com) is at the Physics Department of Universidade Federal de Minas Gerais, Brazil. 
Jonathan Oppenheim (jono@damtp.cam.ac.uk) is in the Dept. of Applied Mathematics and Theoretical Physics, University of Cambridge.  FB is supported by a "Conhecimento Novo" fellowship from Funda\c{c}\~ao 
de Amparo a Pesquisa do Estado de Minas Gerais (FAPEMIG). JO is supported by
the Royal Society, and NSF grant PHY-0551164 while at KITP, Santa Barbara.} 
}
\maketitle

\begin{abstract}

Is there a meaningful quantum counterpart to public communication? We
argue that the symmetric-side channel -- which distributes quantum
information symmetrically between the receiver and the environment -- is
a good candidate for a notion of public quantum communication in
entanglement distillation and quantum error correction.

This connection is partially motivated by [Brand\~ao and Oppenheim, arXiv:1004.3328],
where it was found that if a sender would like to communicate a secret
message to a receiver through an insecure quantum channel using a
shared quantum state as a key, then the insecure quantum channel is only
ever used to simulate a symmetric-side channel, and can always be
replaced by it without altering  the optimal rate. Here we further show, in
complete analogy to the role of public classical communication, 
that assistance by a symmetric-side channel makes equal the distillable
entanglement, the recently-introduced mutual independence, and a
generalization of the latter, which quantifies the extent to which one
of the parties can perform quantum privacy amplification.

Symmetric-side channels, and the closely related erasure channel, have
been recently harnessed to provide examples of superactivation of the
quantum channel capacity. Our findings give new insight into this
non-additivity of the channel capacity and its relation to quantum privacy. In particular, we show that single-copy
superactivation protocols with the erasure channel, which encompasses
all examples of non-additivity of the quantum capacity found to date, can be understood as a conversion of mutual
independence into distillable entanglement.

\end{abstract}

\section{Introduction and Results}
\label{sec:intro}

Suppose two trusted parties, Alice and Bob, and a malicious third party, Eve, share 
noisy classical correlations given by a joint probability distribution $P_{XYZ}$. 
These could come e.g. from measurements on a quantum state shared by them, or from a noisy communication channel, which the eavesdropper is trying to tap. If Alice and Bob's distribution contains some correlations that are 
partially unknown to Eve, they can exploit this to distill a secret-key 
(a shared random variable which is arbitrarily close to being perfectly correlated between
Alice and Bob, and completely unknown to Eve). This process of obtaining key by previously established correlations is known as \textit{information-theoretic} key agreement \cite{Mau93, AC93, CK78, Wyn75}, alluding to the unconditional security of the protocol guaranteed by information-theoretic considerations, and not conditional on any computational-hardness assumption. 

A key resource in this paradigm is \textit{public communication}, which is conveniently represented by a symmetric broadcast channel which delivers the same information to Bob and Eve (or Alice and Eve, if the communication is coming from Bob). The notion of public communication is useful first because it is a realistic one: in many practical situations the parties have access to an authenticated channel which they might freely use to communicate, but which might nonetheless be 
subject to eavesdropping.
Second, the ability to communicate by a public channel turns out to be instrumental in the development of an elegant, and tractable, theory of secret-key distillation \cite{AC93, CK78, Wyn75}. 

For example, in the\textit{ one-way} or \textit{forward} public communication scenario (which is the one considered throughout the rest of this paper), only Alice is able to send public messages to Bob and Eve. In this case the distillable secret-key rate of the distribution $P_{XYZ}$ (when the parties are given infinitely many independent realizations of it) is given by \cite{AC93, CK78, Wyn75}
\begin{equation} \label{distillalesecretkeyclassical}
C(P_{XYZ})=\sup_{X \rightarrow V \rightarrow U}I(V:Y|U)-I(V:Z|U),
\end{equation}
with the conditional mutual information $I(V:Y|U):=H(VU)+H(YU)-$$H(VYU)-H(U)$, the Shannon entropy $H(X):=-\sum_x P_{X=x}\log P_{X=x}$, and the supremum taken over the Markov chain $X\rightarrow V\rightarrow U$. 

The formula in Eq. (\ref{distillalesecretkeyclassical})  is so-called \textit{single-letter}, meaning that an optimization over a single copy of the probability distribution gives the asymptotic rate. Moreover it is \textit{additive}, i.e. for two probability distributions $P_{XYZ}$ and $Q_{X'Y'Z'}$, $C(P_{XYZ}\otimes Q_{X'Y'Z'}) = C(P_{XYZ}) + C(Q_{X'Y'Z'})$ \cite{AC93}. We can then say that Eq. (\ref{distillalesecretkeyclassical})
completely characterizes how to optimally distill secret-key in the one-way scenario. 
In contrast, if one instead consider the task of distilling a key from common randomness without any communication, one finds a much more complicated theory, in which even the determination of which probability distributions allow the extraction of key remains an open problem, let alone the derivation of a tractable formula for the distillable secret-key rate.

In quantum information theory, the paradigm described above has two natural analogues, and both have been extensively analysed \cite{HHHH09, DW05, HHHO09}. The first is to distill a secret-key from a tripartite quantum state $\ket{\psi_{ABE}}$ shared by Alice, Bob and Eve \cite{HHHO09}. Alice and Bob can perform any operation allowed by quantum mechanics on their shares of the state, while (in the one-way setting considered here) Alice can communicate public \textit{classical} messages to Bob and Eve. The second is entanglement distillation \cite{HHHH09}, in which Alice and Bob  wish to distill Einstein-Podolsky-Rosen (EPR) pairs from a shared state $\psi_{AB}$ by local quantum operations and, again, \textit{classical} communication from Alice to Bob (here too, although not needed, one can consider that an eavesdropper has a purification of $\psi_{AB}$ i.e.\ a pure
state $\psi_{ABE}$ such that $\psi_{AB}=\tr_E{\psi_{AB}}$, and Eve 
learns all the classical communication that Alice sends to Bob). 

In both paradigms, the shared randomness is extended from the original classical probability distribution to a quantum state. The public communication, however, remains the same; even in the quantum case only classical messages can be publicly communicated. A natural question then emerges: is there a meaningful notion of \textit{public quantum communication}?

A first objection to such a notion comes from the fact that quantum information cannot be copied \cite{WZ82, Die82} and, hence, it is problematic to consider a procedure which gives the same quantum information to several parties, as in classical public communication. A second objection comes from practical considerations. Both secret-key and entanglement distillation are important operational primitives, and the paradigm of quantum local operations and public classical communication emerges  naturally from the resources which are usually available, e.g. in a quantum key-distribution set-up. So if we do not have an interesting setting where a notion of public quantum communication is needed, why bother with such concept?
In this paper we show that, at least in the one-way setting, both objections  are not well founded, and that there is room for a useful definition of quantum public communication. 

\subsection{Symmetric-side Channels}
\label{ss:symmetric-side}

In the same way a broadcast channel (which sends the same information to the two receivers) is employed as a model of a classical public communication channel, we will use a quantum \textit{symmetric-side} channel \cite{SSW08} as a model of quantum public communication. This is the channel which maps quantum information symmetrically between the receiver and the environment (see section \ref{backdef} for a precise definition). 

Note that although both the receiver and the eavesdropper get the same quantum information in the symmetric-side channel, neither of them get the full information originally contained in the sender's state, so there is no cloning of quantum information. Moreover, if the sender prepares, as the input of the symmetric-side channel, a state diagonal in the computational basis, then the channel has the same effect as a classical symmetric  broadcast channel. In this way symmetric-side channels can be seen as at least as \textit{a} natural generalization of public communication to quantum states.

Symmetric-side channels were introduced by Smith, Smolin and Winter \cite{SSW08} with the goal of obtaining a more tractable upper bound on the quantum capacity of quantum channels, defined as the optimal asymptotic rate at which a quantum channel can transmit  quantum states faithfully. They analysed how assistance by a symmetric-side channel could improve the quantum channel capacity and the (one-way) distillable entanglement -- given by the maximum rate of EPR pairs that can be asymptotically extracted from a quantum state by local operation and classical communication (from Alice to Bob). Remarkably, they found single-letter, additive expressions for both assisted capacities: for a bipartite state $\psi_{AB}$ with purification $\ket{\psi_{ABE}}$, the symmetric-side channel assisted distillable entanglement (see section \ref{backdef} for a formal definition)  
can be manipulated into the form      
\begin{equation} \label{ssdistillableent}
D_{ss}(\psi_{AB}) =\sup_{A\rightarrow a\alpha}\frac{1}{2}
(I(a:B|\alpha)-I(a:E|\alpha)),
\end{equation}
using entropic identities
and with $I(a:B|\alpha):=S(a\alpha)+S(B\alpha)-S(aB\alpha)-S(\alpha)$ the conditional mutual information, $S(a):=-\tr \rho_a\log \rho_a$ the von Neumann entropy, and the supremum
taken over all channels which maps $A$ to $a \alpha$. 
This expression is a direct quantum generalization of Eq. (\ref{distillalesecretkeyclassical}) and in itself already suggests a formal analogy of distillable entanglement and symmetric-side channels with classical secret-key and public communication.

\subsection{An Operational Motivation}
\label{ss:operational}

For an operational motivation for the view of symmetric-side channels as public quantum communication, we consider the quantum \textit{one-time-pad} problem analysed and solved in \cite{BO10} (partially employing the techniques we developed in this paper). The setting is as follows. Alice would like to send to Bob secret classical or quantum messages, using an ideal, but insecure, quantum channel which might be intercepted by an eavesdropper, who should not be able to learn anything about the message being sent. 

Alice and Bob can make use of the insecure channel for secure communication if they share in addition a \textit{secret-key}. Then using their secret correlations Alice can encode the message in a way that  (i) Bob can decode it in the case that Eve does not intercept the states
sent down the insecure quantum channel and (ii) Eve cannot distinguish the different
messages if she intercepts the sent states. We assume that the key is given by (several copies of) a quantum state $\ket{\psi_{ABE}}$ shared by Alice, Bob and Eve and the question is to find out what is the optimal rate at which the state can be used to encrypt classical or quantum messages.

This problem was first considered in the noiseless case, in which Alice and Bob share perfect classical key or EPR pairs \cite{BR03, AMTW00, Leu00}. In Ref. \cite{SW06}, in turn, Schumacher and Westmoreland analysed the case in which the key shared by Alice and Bob is a mixed bipartite quantum state $\psi_{AB}$, which is not correlated  with the eavesdropper. Interestingly, they found the optimal rate at which the state can be used as a one-time-pad for classical messages  to be given by the quantum mutual information of $\psi_{AB}$: $I(A : B)_{\psi} = S(A)_{\psi} + S(B)_{\psi} - S(AB)_{\psi}$. 

In \cite{BO10} we considered the general case, in which Alice and Bob have an arbitrary quantum state, in general correlated with Eve. We found that the optimal rate at which the state can be used as a one-time-pad for quantum information turns out to be given by Eq. (\ref{ssdistillableent}). It is intriguing that it is the symmetric-side channel assisted distillable entanglement that appears as the optimal rate, even though the problem makes no mention in any way of the symmetric-side channel. 

The proof of our result reveals an interesting aspect of this task: the insecure quantum channel is only ever used to simulate a symmetric-side channel, meaning that in the optimal protocol Alice first locally simulates a symmetric-side channel, sends through the insecure channel the output part of the symmetric-side channel which would go to Bob, and traces out the part that would go to Eve. It thus follows that there is no difference if Alice and Bob are connected by an insecure ideal channel or a symmetric-side channel! 

We can therefore consider the quantum one-time-pad as an operational setting where the idea of a symmetric-side channel as public quantum communication naturally appears (though in an indirect manner).

\subsection{Superactivation of the Channel Capacity}
\label{ss:superactivation}

There is another line of investigation in which symmetric-side channels have been
shown very useful: in exhibiting examples of non-additivity of the quantum channel capacity \cite{SY08}. By the no-cloning theorem \cite{WZ82, Die82}, the symmetric-side channel can be seen to have zero quantum capacity. However, in \cite{SY08} Smith and Yard noted that a consequence of Eq. (\ref{ssdistillableent}) and the formula of \cite{DW05} for the one-way distillable secret-key rate ($K_{\rightarrow}$) is 
\begin{equation} \label{ssdistilableveruskey}
D_{ss}(\psi_{AB}) \geq K_{\rightarrow}(\psi_{AB})/2,
\end{equation} 
for all bipartite states $\psi_{AB}$.
The equation above is striking because there are examples of states for which the distillable entanglement is zero, but the distillable secret-key is not \cite{HHHO05, HPHH08, HHHO09} . In this way, and by considering quantum channels which generate such states, we find an example of two quantum channels (the symmetric-side channel and the other channel which can only produce states with zero distillable entanglement, but some with positive distillable key) each with zero quantum capacity, but whose tensor product has positive quantum capacity. This effect has been termed the \textit{superactivation} of the quantum capacity. 

Equation (\ref{ssdistilableveruskey}) shows a curious property of the symmetric-side channel: it allows the conversion of secret-key into EPR pairs (at half the rate). An interesting question, raised already in \cite{SY08} and further explored in \cite{SS08, LWZG09, SS09}, asks whether there is a more fundamental relation between entanglement and secrecy in the presence of symmetric-side channels. For instance, might the distillable entanglement and distillable secret-key, when assisted by symmetric-side channels, become the same? Although it is rather unlikely that this is the case (see the remark after the proof of Theorem \ref{theoallthesame}), here we will show that a relaxed version of the statement \textit{is} true. That brings us to the final concept that we will touch in this work.

\subsection{Mutual Independence}
\label{ss:mutualindep}

The definition of secret-key consists of two requirements: (i) Alice and Bob systems should be 
classical, and perfectly correlated and (ii) their state should not be correlated in any way with the eavesdropper.  A relaxed and fully quantum definition of private correlations has recently
been introduced~\cite{HOW09}, in which only the second requirement 
is kept. Then given a bipartite quantum state $\psi_{AB}$, the degree of (potentially noisy) private correlations of Alice and Bob, termed \textit{mutual independence} ($I_{\text{ind}}(\psi_{AB})$), is given by (half) the mutual information of a state extracted by Alice and Bob  which is product with Eve's state, who is assumed to hold a purifying state for $\psi_{AB}$.  Their actions are on asymptotically many copies of $\psi_{AB}$
and one can consider the mutual independence under different types of
operations e.g. by local operations and classical communication.

An operational significance of this new quantity was given in \cite{HOW09}: the sum of quantum communication by Alice and Bob required to send their state to a receiver, in the presence of free entanglement, is given in terms of
the mutual independence rate with no classical communication. 
Given the view of mutual independence as a more relaxed form of private correlations than secret-key, we might ask whether a similar relation as in Eq. (\ref{ssdistilableveruskey}) holds. We will show that this is indeed the case and that the relation turns out to be actually stronger than with secret-key .

\subsection{Our Results}
\label{ss:results}

Our first contribution is to introduce an even more relaxed notion of private correlations, which we call \textit{weak mutual independence} (see section \ref{backdef}). Its definition is almost the same as that of mutual independence, but here we only require that Alice's state is completely decoupled from Eve's. In the setting where no classical communication is allowed, the optimal protocol is just for Alice to split her system in two and trace out one of them, making herself product with Eve and at the same time trying to retain as much mutual information as possible with Bob. We can then see this quantity as a measure of Alice's ability to perform quantum privacy amplification against the eavesdropper. 

In section \ref{oneformulafitsall} we derive an entropic capacity formula (alas, a regularized one) for weak mutual independence in the zero and one-way classical communication cases, something that remains an open question for mutual independence. Our formula turns out to be a direct generalization of both the formulae for (one-way) distillable entanglement and distillable secret-key of Devetak and Winter \cite{DW05}. 

Our main result, presented in section \ref{capacitiesareeuqal}, is the following: when assisted by a symmetric-side channel, the weak mutual independence rate $(W_{\text{ind}, ss}$), the mutual independence rate ($I_{\text{ind}, ss}$), and the distillable entanglement ($D_{ss}$) become the same, i.e. for every state $\psi_{AB}$
\begin{equation} \label{allareequalintro}
W_{\text{ind}, ss}(\psi_{AB}) = I_{\text{ind}, ss}(\psi_{AB}) = D_{ss}(\psi_{AB}).
\end{equation}

We note that an analogous equation holds true classically, if we replace distillable entanglement by distillable secret-key and redefine the two mutual independence rates removing the half factor presented in the quantum case. Indeed, the rate of weak mutual independence which can be attained classically is at least as large as Eq. (\ref{ssdistillableent}), as one can even get secret-key at this rate. But it also holds that the weak mutual independence rate cannot be larger than Eq. (\ref{ssdistillableent}). An optimal protocol for weak mutual independence by public communication consists of two steps: Alice first applies a transformation to her random variable (given by $n$ realizations of the distribution $P_{XYZ}$) $X_n \rightarrow V_n$ obtaining $V_n$ and then communicate part of it to Bob and Eve, which in turn we can model as an application of a map $V_{n} \rightarrow U_{n}$ and the communication of the random variable $U_{n}$. As the protocol extracts mutual independence, we must have that, asymptotically, $I(V_{n} : Z_{n}U_{n}) \rightarrow 0$, since Alice's final random variable must be decoupled from Eve's. Therefore the weak mutual independence rate is bounded as follows
\begin{eqnarray}
\frac{1}{n} I(V_n : Y_n U_n) &\lesssim& \frac{1}{n} I(V_n : Y_n U_n) - \frac{1}{n}I(V_n : Z_nU_n) \nonumber \\ &=& \frac{1}{n} I(V_n : Y_n |U_n) - \frac{1}{n}I(V_n : Z_n|U_n) \nonumber \\ &\leq& C(P_{XYZ}),
\end{eqnarray}
where the equality in the second line follows from a simple entropic manipulation and the inequality in the last line from the additivity of $C(P_{XYZ})$.    

Eq. (\ref{allareequalintro}) is also key for our result in \cite{BO10} on the quantum one-time-pad. A protocol for sending classical messages would be to first distill mutual independence using the insecure channel to simulate a symmetric-side channel.  One then is
in the situation considered by Schumacher and Westmoreland \cite{SW06}, where initially Alice 
and Bob are decoupled from Eve, and one can then implement their protocol.  
That the information that goes through the insecure channel in the first part of the protocol does not leak information to Eve can be seen as follows: Alice locally simulates the symmetric
side channel which she would have used to extract mutual independence, and then sends Bob's output to him through the insecure
quantum channel and traces out the part that would go to Eve.  If Eve intercepts the state, then
because of the symmetry of Bob and Eve in the symmetric-side channel, she is just getting
the information which would have anyway gotten to her in the case where there was actually
a symmetric side channel, and therefore, she has to be decoupled from Alice and Bob's final state.

Finally, Eq. (\ref{allareequalintro}) is also interesting in the context of superactivation of the quantum capacity or distillable entanglement. For one thing, it shows that when looking for more superactivation protocols with the symmetric-side channel, one can focus on the rather indiscriminate task of making part of Alice's state product with the environment. 
In section \ref{superactivation} we show another connection of superactivation with mutual independence, which relates the weak mutual independence rate \textit{without} assistance by any side channel, with the  maximum coherent information (whose regularization gives the distillable entanglement) achievable with the assistance of an erasure channel. This is a channel obtained by a particular encoding of the symmetric-side channel and is the one actually used in all concrete examples of non-additivity found so far \cite{SY08, SS08, LWZG09, SS09, Opp09}
. This suggests that weak mutual independence might not be increased by symmetric-side channels, which would be a considerable improvement of Eq. (\ref{allareequalintro}), but which we leave as an open problem.

\section{Background and Definitions} \label{backdef}

The \textit{symmetric-side channel} is defined by the property that it maps the input state symmetrically to the receiver and the environment. For completeness,
and in detail, we follow Ref. \cite{SSW08}: consider the symmetric subspace ${\cal S}_d$ between two $d$-dimensional Hilbert spaces spanned by the basis:
\begin{equation}
\ket{(i, j)} := 
\begin{cases}
\frac{1}{\sqrt{2}}(\ket{i, j} + \ket{j, i}) & \text{for} \hspace{0.3 cm} i < j \\
\ket{i, i} & \text{for} \hspace{0.3 cm} i = j.
\end{cases}
\end{equation}
Let $U_{ss, d} : A \hookrightarrow BE$, with $A \cong \mathbb{C}^{d(d+1)/2}$ and $BE \cong {\cal S}_d$, be an isometry which maps a basis of the $d(d+1)/2$-dimensional Hilbert space into the 
$\ket{(i, j)}$, in some order. Then the $d$-dimensional symmetric-side channel is defined as
\begin{equation}
\Lambda_{ss, d}(\rho) := \tr_{E} \hspace{0.1 cm} U_{ss, d} \rho U_{ss, d}^{\cal y}.
\end{equation}

In the (fifty-fifty) \textit{erasure channel}, with probability half the quantum information is sent to the receiver intact, while with probability half the information is completely lost to the environment and the receiver gets an error flag. Let $U_{e, d}\ : A \hookrightarrow BE$, with $A \cong\mathbb{C}^d$ and $B, E \cong \mathbb{C}^{d+1}$, be the isometry defined as
\begin{equation}
U_{e, d} \ket{i} = \frac{1}{2}(\ket{i, e} + \ket{e, i})
\end{equation}
for all $i \in \{0, ..., d-1 \}$, with $\ket{e} \equiv \ket{d}$. Then 
the erasure channel is given by
\begin{eqnarray}
\Lambda_{e, d}(\rho) &:=& \tr_{E} \hspace{0.1 cm} U_{e, d} \rho U_{e, d}^{\cal y} \nonumber \\ &=& \frac{1}{2}\rho +\frac{1}{2}\ket{e}\bra{e}.
\end{eqnarray}

For a channel $\Lambda(\rho) := \tr_E(U \rho U^{\cal y})$ we define its complementary channel as $\Lambda^{c}(\rho) := \tr_B(U \rho U^{\cal y})$. We say $\Lambda$ is \textit{anti-degradable} if there is another quantum operation ${\cal E}$ such that $\Lambda = {\cal E} \circ \Lambda^{c}$, i.e. Eve can simulate the channel from Alice to Bob by applying the operation ${\cal E}$ \cite{DS05}. Both the symmetric-side channel and the erasure channel are examples of anti-degradable channels and all such channels have zero quantum capacity \cite{DS05}.   

Note that one can use a $(d+1)$-dimensional symmetric-side channel to simulate a $d$-dimensional erasure channel in a very simple way: the sender only have to encode the $d$-dimensional input space into the subspace which gets mapped by $U_d$ to the subspace $\{ \ket{(i, j)} : i \in \{ 0, ..., d-1 \}, j = d \}$. In fact this is a particular case of a more general property of the symmetric-side channel, which can be used to simulate any other anti-degradable channel by appropriate encoding and decoding operations.

\begin{theorem} \label{ssanti}
For every anti-degradable channel $\Lambda$ there is an integer $d$ and quantum operations ${\cal E}$ and ${\cal F}$ such that  
\begin{equation}
\Lambda = {\cal E} \circ \Lambda_{ss, d} \circ {\cal F}.
\end{equation}
\end{theorem}
\noindent The proof of the theorem is given in Appendix \ref{app1}. 

We now turn to the four main quantities which we will be concerned with: distillable entanglement \cite{DW05}, secret key \cite{DW05}, mutual independence \cite{HOW09}, and weak mutual
independence. They are all given by the optimization of a certain cost
function over a restricted set of operations. Here we are interested in the following classes of operations:

\begin{itemize}
\item 
Local operations (by Alice and Bob), without communication. We denote this class by $\emptyset$. 
\item Local operations and one-way classical communication from Alice to Bob. The class will be denoted by $\rightarrow$.
\item Local operations and forward communication by (unlimited many uses of) an erasure channel. The class is denoted by $e$.
\item Local operations and forward communication by (unlimited many uses of) a symmetric-side channel. The class is denoted by $ss$.
\end{itemize}
In the following let ${\cal C}$ be one of the class of operations defined above. 

\textit{Distillable Entanglement:} Given a (mixed) bipartite state $\psi_{AB}$, a ${\cal C}$-protocol for entanglement distillation is formed by a sequence of maps $\Lambda^{(n)}$ from ${\cal C}$ such that 
\begin{equation}
\lim_{n \rightarrow \infty}\Vert \Lambda^{(n)}(\psi_{AB}^{\otimes n}) - \Phi(M_n) \Vert_1 = 0,  
\end{equation}
where $\Phi(M_n)$ is the $M_n$-dimensional maximally entangled state given by 
\begin{equation}
\Phi(M_n) = \frac{1}{M_n}\sum_{i=0}^{M_n-1}\sum_{j=0}^{M_n-1}\ket{i,i}\bra{j,j}.
\end{equation}

\begin{definition}
(distillable entanglement) Given a state $\psi_{AB}$ and an entanglement distillation ${\cal C}$-protocol ${\cal P} = \Lambda^{(n)}$, define the rate
\begin{equation}
R({\cal P}, \psi_{AB}) := \liminf_{n \rightarrow \infty} \frac{\log M_n}{n}.
\end{equation}
The ${\cal C}$-distillable entanglement of $\psi_{AB}$ is given by
\begin{equation}
D_{\cal C}(\psi_{AB}) := \sup_{\cal P} R({\cal P}, \psi_{AB}),
\end{equation}
\end{definition}
\vspace{0.2 cm}

\textit{Distillable Secret-Key:} A  ${\cal C}$-protocol for secret-key distillation is a sequence of maps $\Lambda^{(n)}$ from ${\cal C}$ with the property that  
\begin{equation} 
\rho^{(n)}_{ABE} := \Lambda^{(n)}\otimes \id_{E}(\psi_{ABE}^{\otimes n}),
\end{equation}
with $\ket{\psi_{ABE}}$ a purification of $\psi_{AB}$, is such that
\begin{equation}
\lim_{n \rightarrow \infty} \left \Vert \rho^{(n)}_{ABE} - \frac{1}{M_n}\sum_{i=0}^{M_n-1}\ket{i,i}\bra{i,i} \otimes \rho^{(n)}_{E} \right \Vert_1 = 0.
\end{equation}

\begin{definition}
(distillable secret-key) Given a state $\psi_{AB}$, consider a ${\cal C}$-protocol for key distillation ${\cal P} = \Lambda^{(n)}$ and define the rate
\begin{equation}
R({\cal P}, \psi_{AB}) := \liminf_{n \rightarrow \infty} \frac{\log M_n}{n}.
\end{equation}
The ${\cal C}$-distillable secret-key of $\rho_{AB}$ is given by
\begin{equation}
K_{\cal C}(\psi_{AB}) := \sup_{\cal P} R({\cal P}, \psi_{AB}),
\end{equation}
\end{definition}
Note that distillation of secret key via public communication is completely equivalent
to distilling a class of states $\gamma_n$ called {\it pbits} via local operations
and classical communication \cite{HHHO05}. The states $\gamma_n$ are a broader class of states than pure entangled ebits, which allow Alice and Bob to get (almost) perfect classical correlations unknown to an adversary, who has a purifying system of their state.
\vspace{0.2 cm}

\textit{Mutual Independence:} Following \cite{HOW09}, we call a ${\cal C}$-protocol for extracting mutual independence from $\psi_{AB}$ any 
sequence of maps $\Lambda^{(n)}$ from the class ${\cal C}$ with the property that
\begin{equation} 
\label{defweaknew}
\rho^{(n)}_{ABE} := \Lambda^{(n)}\otimes \id_{E}(\psi_{ABE}^{\otimes n})
\end{equation}
is such that
\begin{equation}
\lim_{n \rightarrow \infty} \Vert \rho^{(n)}_{ABE} - \rho^{(n)}_{AB} \otimes \rho^{(n)}_{E} \Vert_1 = 0.
\end{equation}

\textit{Weak Mutual Independence:} We define a new quantity, which is a weaker notion of mutual independence where we only require Alice to be product with Eve. We call a
protocol for extracting weak mutual independence any sequence of maps $\Lambda^{(n)}$from ${\cal C}$ such that
\begin{equation}
\lim_{n \rightarrow \infty}\Vert \rho^{(n)}_{AE} - \rho^{(n)}_{A} \otimes \rho^{(n)}_{E} \Vert_1 = 0.
\end{equation}
with $\rho^{(n)}_{ABE}$ defined in Eq. (\ref{defweaknew}) above. 

\begin{definition} \label{mutulandweakmutual}
(mutual independence and weak mutual independence) Given a state $\psi_{AB}$, consider a ${\cal C}$-protocol for extracting 
(weak) mutual independence ${\cal P} = \Lambda^{(n)}$. Define the rate
\begin{equation}
R({\cal P}, \psi_{AB}) := \liminf_{n \rightarrow \infty} \frac{1}{2n} I(A : B)_{\Lambda^{(n)}(\psi_{AB}^{\otimes n})}.
\end{equation}
Then we define the ${\cal C}$-mutual independence rate of $\psi_{AB}$ as
\begin{equation}
I_{\text{ind}, {\cal C}}(\psi_{AB}) := \sup_{\cal P} R({\cal P}, \psi_{AB}),
\end{equation}
while the ${\cal C}$-weak mutual independence rate of $\rho_{AB}$ is defined as
\begin{equation}
W_{\text{ind}, {\cal C}}(\psi_{AB}) := \sup_{\cal P} R({\cal P}, \psi_{AB}),
\end{equation}
\end{definition}

\section{One formula fits all} \label{oneformulafitsall}

We have seen that mutual independence and its weak variant can be seen as extensions of distillable entanglement, or distillable secret key, to a setting in which the condition that Alice and Bob's systems are perfectly correlated (either quantumly as ebits or classically as pbits \cite{HHHO09}) is dropped, and only the privacy condition that Eve's state is factored out is required.
In this section we will derive a capacity formula for weak mutual independence. This turns out to be a generalization of both the formula for one-way distillable entanglement and one-way distillable secret-key derived by Devetak and Winter \cite{DW05}. 

Consider the following classes of completely positive (CP) maps:
\begin{itemize}
\item $RO$ (\textit{rank-one CP maps}): all maps $\Lambda$ of the form $\Lambda(\sigma) = A \sigma A^{\cal y}$.
\item $QC$ (\textit{quantum-to-classical maps}): all maps $\Lambda$ of the form $\Lambda(\sigma) = \sum_j \tr(A_j \sigma)\ket{j}\bra{j}$, with $A_j\geq 0$ and $\sum_j A_j \leq \id$.
\item $CP$ (\textit{general CP maps}): this is the class formed by all CP maps.
\end{itemize}
The next theorem shows that (one-way) distillable entanglement, distillable secret key, and weak mutual independence are given by the same quantity optimized over the three classes of operations defined above, respectively.

\begin{theorem}
For a pure state $\ket{\psi_{ABE}}$ and a class of operations ${\cal C}$ we define
\begin{equation}
G_{{\cal C}}(\psi_{AB}) :=  \max_{\rho \in {\cal C}} \frac{1}{2} \left(I(a : B|\alpha)_{\rho} - I(a : E |\overline{\alpha})_{\rho} \right)
\end{equation}
with
\begin{multline*}
    \rho_{aB\alpha E \overline{\alpha}}  := \\
       \sum_k p_k ({\cal E}_k \otimes \id_{BE})(\psi_{ABE}) \otimes \ket{k, k}_{\alpha \overline{\alpha}}\bra{k, k}, 
\end{multline*}
and the maximization taken over all sets of CP maps $\{ {\cal E}_k : A \hookrightarrow a \}_k$ contained in the class ${\cal C}$ and whose elements sum up to a quantum operation ${\cal\ E} := \sum_k p_k {\cal E}_k$. Let 
\begin{equation}
G_{\cal C}^{\infty}(\psi_{AB}) := \lim_{n \rightarrow \infty} \frac{1}{n}G_{\cal C}(\psi_{AB}^{\otimes n}).
\end{equation}

Then
\begin{equation} \label{DistillableEntRO}
 D_{\rightarrow}(\psi_{AB}) = {G}_{RO}^{\infty}(\psi_{AB}),
\end{equation}
\begin{equation} \label{DistillableKeyQC}
K_{\rightarrow}(\psi_{AB}) = 2G_{QC}^{\infty}(\psi_{AB}),
\end{equation}
and
\begin{equation} \label{weakmutualCP}
W_{\text{ind}, \rightarrow}(\psi_{AB}) = G_{CP}^{\infty}(\psi_{AB}).
\end{equation}
\end{theorem}

\begin{proof}
Equation (\ref{DistillableKeyQC}) follows directly from Ref. \cite{DW05}, while Eq. (\ref{DistillableEntRO}) is a simple rearrangement of the formula found in Ref. \cite{DW05}. We only have to note that because each $\psi_{k, ABE} := ({\cal E}_k \otimes \id_{BE})(\psi_{ABE})$ is a pure state, it follows that
\begin{eqnarray}
&& G_{RO}(\psi_{AB})  \\ &=& \lim_{n \rightarrow \infty} \sum_k p_k \frac{1}{2}(I(a:B)_{\psi_k} - I(a:E)_{\psi_k}) \nonumber \\ &=& \lim_{n \rightarrow \infty}\sum_k p_k I(a\rangle B)_{\psi_k} \nonumber
\end{eqnarray}
which is Devetak-Winter formula for the one-way distillable entanglement.
 
So it remains to prove Eq. (\ref{weakmutualCP}). Let us start showing the achievability of $G_{CP}(\psi_{AB})$, which by block coding implies that $G_{CP}^{\infty}(\psi_{AB})$ is achievable as well. The protocol has three steps. In the first, Alice applies the operation
\begin{equation}
{\cal F}_(\sigma) = \sum_k p_k {\cal E}_k(\sigma) \otimes \ket{k}\bra{k}_X,
\end{equation}
to $n$ copies of her share of the state, communicates the classical information in the register $X$ to Bob and Eve and traces out $X$, obtaining $n$ copies of the state 
\begin{multline}\label{stateafterccW}
\phi_{aB\alpha E \overline{\alpha}} \\= \sum_k p_k ({\cal E}_k \otimes \id_{BE})(\psi_{ABE})\otimes \ket{k,k}_{\alpha \overline{\alpha}}\bra{k,k},
\end{multline}
where $\alpha$ and $\overline{\alpha}$ are held by Bob and Eve, respectively.

In a second step Alice projects her system into its \textit{typical subspace} \cite{Sch95}, outputing an error flag  when the projection fails and getting the state $\phi_{aB\alpha E \overline{\alpha}}^n$.

Finally, Alice splits her $a_n$ (which labels the system held by Alice) into two registers $a_{1, n}$ and $a_{2, n}$ of size 
\begin{equation} 
\label{sizea2}
\lim_{n \rightarrow \infty}\frac{1}{n}\log |a_{2, n}| = \lim_{n \rightarrow \infty} \frac{1}{2n} I(a : E\overline{\alpha})_{\phi_{aE\overline{\alpha}}}
\end{equation}
in such a way that 
\begin{equation}
\lim_{n \rightarrow \infty} \Vert \phi_{a_{1, n}  E\overline{\alpha}} - \phi_{a_{1, n}} \otimes \phi_{E\overline{\alpha}} \Vert_1 = 0. 
\end{equation}
That such a splitting always exists is shown in Lemma \ref{decoupling}. As Alice's final state is product with Eve's, the protocol extracts weak mutual independence. The rate is given by
\begin{eqnarray} 
&&\lim_{n \rightarrow \infty} \frac{1}{2n} I(a_{1, n} : B)_{\phi} \\   
&=&\lim_{n \rightarrow \infty}\left(\frac{1}{2n} I(a_{1, n} a_{2, n} : B)_{\phi} - \frac{1}{2n} I(a_{2, n} : B | a_{1, n})_{\phi} \right), \nonumber
\end{eqnarray}
where we used the chain rule. From the bound
\begin{equation} \label{boundcondmutualinfo}
I(a_{2, n} : B | a_{1, n}) \leq 2 \log(|a_{2, n}|)
\end{equation}
and Eq. (\ref{sizea2}) we then get
\begin{eqnarray}
&&\lim_{n \rightarrow \infty} \frac{1}{2n} I(a_{1, n} : B)_{\phi} \\ & \geq&  \lim_{n \rightarrow \infty} \left( \frac{1}{2n} I(a : B)_{\phi} - \frac{1}{2n} I(a : E)_{\phi} \right) \nonumber,
\end{eqnarray}
which shows that $W_{\text{ind}, \rightarrow}(\psi_{AB}) \geq G_{CP}(\psi_{AB})$. 

The converse follows almost directly from the definition of $W_{\text{ind}, \rightarrow}$. Consider an optimal protocol for extracting weak mutual independence as in Def. (\ref{mutulandweakmutual}). The optimal 1-way LOCC operations $\Lambda^{(n)}$ can be assumed to have the form 
\begin{multline}
\Lambda^{(n)}(\sigma) = \\ \sum_k q_{k, n} ({\Lambda}_{k, n}\otimes \id_{BE})(\sigma)\otimes \ket{k,k}_{\alpha \overline{\alpha}}\bra{k,k}
\end{multline}
for a quantum instrument $\{ q_{k, n}, \Lambda_{k, n} \}$ (i.e.\ a set of CP maps implemented
with probability $q_{k,n}$). This is so because any action of Bob would only decrease his mutual information with Alice. Thus the optimal one-way LOCC protocol for $W_{{\text ind}, \rightarrow}$ consists of Alice applying an instrument to her system and communicating which CP map was implemented to Bob and Eve. Then
\begin{eqnarray}
&&W_{\text{ind}, \rightarrow}(\psi_{AB})  \\ 
&=& \liminf_{n \rightarrow \infty} \frac{1}{2n} I(A : B \alpha)_{\Lambda^{(n)}(\psi_{AB}^{\otimes n})}   \nonumber \\ 
&=& \liminf_{n \rightarrow \infty} \frac{1}{2n} \left( I(A : B\alpha)_{\Lambda^{(n)}(\psi_{AB}^{\otimes n})} - I(A : E\alpha)_{\Lambda^{(n)}(\psi_{AE}^{\otimes n})}\right) \nonumber \\
&\geq& G_{CP}^{\infty}(\psi_{AB}), \nonumber
\end{eqnarray}
where the the second equality follows from the fact that Alice's state is asymptotically product with Eve's, and the last equality comes from the definition of  $G_{CP}^{\infty}$.
\end{proof}

\vspace{0.2 cm}
\textbf{Remark:} Following the proof of the theorem it is also straightforward to derive a formula for the zero-way weak mutual independence:
\begin{equation}
W_{\text{ind}, \emptyset}(\psi_{AB}) = \lim_{n \rightarrow \infty} \frac{1}{n}W_{\text{ind}, \emptyset}^{(1)}(\psi_{AB}^{\otimes n})
\end{equation}
with
\begin{equation}
W_{\text{ind}, \emptyset}^{(1)}(\psi_{AB}) := \max_{A \hookrightarrow a \alpha} \frac{1}{2} \left( I(a : B) - I(a : E) \right), 
\end{equation}
where the maximization is taken over all isometries mapping $A$ to $a \alpha$. 

Also along very similar lines to the proof above, we get the following generalization of the theorem: for a channel $\Lambda$ define the quantity 
\begin{equation}
G_{\Lambda}(\psi_{AB}) := \lim_{n \rightarrow \infty} \lim_{m \rightarrow \infty} \frac{G^{(1)}(\Lambda^{\otimes m}, \psi_{AB}^{\otimes n})}{n}
\end{equation}
with
\begin{eqnarray}
&& G^{(1)}(\Lambda, \psi_{AB})  \\ 
&:=&\sup_{A\rightarrow a\alpha_1\alpha_2}\frac{1}{2}
\left[I(a:\Lambda(\alpha_1)B)_{\phi}-I(a:\Lambda^c(\alpha_1)E)_{\phi}\right] \nonumber
\end{eqnarray}
where $\Lambda^c$ is the conjugate channel of $\Lambda$, the optimization is taken over all isometries mapping $\ket{\psi}_{ABE}$ to $\ket{\psi}_{a \alpha_1 \alpha_2 B E}$, and the $''\Lambda(\alpha_1)''$ in the formula  is a shorthand for the state obtained by applying $\Lambda$ to the $\alpha_1$ register of $\ket{\phi}_{a \alpha_1 \alpha_2 BE}$. Then the distillable entanglement, secret-key capacity, and weak mutual independence capacity (zero-way, one-way, assisted by erasure channel, by a symmetric-side channel, or by any other channel) are
all specific cases of this formula for particular choices of the channel $\Lambda$.  

The following Lemma in proved in \cite{ADHW08}:

\begin{lemma} \label{decoupling}
(\cite{ADHW08} Decoupling Lemma) 
For every bipartite state $\psi_{AE}$ let $\psi_{A^nE^n}$ be defined as 
\begin{equation}
\ket{\psi_{A^nE^n}} := P_{n, \varepsilon_n} \ket{\psi_{AE}}^{\otimes n} / \Vert P_{n, \varepsilon_n} \ket{\psi_{AE}}^{\otimes n} \Vert,
\end{equation}
with $P_{n, \varepsilon_n}$ the projector onto the $\varepsilon_n$-typical subspace of $\ket{\psi_{AE}}^{\otimes n}$. Then for every sequence $\{ \varepsilon_n \}$ going to zero, there is a sequence of isometries $V_{n} : A_n \hookrightarrow A_{1, n} A_{2, n}$ such that 
\begin{equation}
\lim_{n \rightarrow \infty} \left \Vert \tr_{A_{1, n}} \left(V_{n}\psi_{AE}^{\otimes n} V_{n}^{\cal y}\right) - \tau_n\otimes\psi_E^{\otimes n} \right \Vert_1 = 0,
\end{equation}
with $\tau_n$ the maximally mixed state in $A_{2, n}$, and
\begin{equation}
\lim_{n \rightarrow \infty} \frac{ \log|A_{1, n}|}{n} = I(A : E)_{\psi}
\end{equation}
\end{lemma}


\section{Assisted Capacities Are Equal} \label{capacitiesareeuqal}

In this section we prove that under the assistance of an erasure channel or a symmetric-side channel, distillable entanglement, mutual independence, and weak-mutual independence become the same. This result is used in Ref. \cite{BO10} to show that these symmetric-side channel assisted capacities give the optimal rate at which a quantum state shared by Alice and Bob can be used as a one-time-pad to encrypt quantum messages which are sent down an insecure quantum channel. 

In the following lemma we give a characterization of states with perfect weak mutual independence analogous to the characterization of \cite{HHHO05} for pbits and of \cite{HOW09} for states with perfect mutual independence. We say a state $\ket{\psi_{a \alpha B E}}$ has perfect weak mutual independence if $\psi_{aE} = \psi_a \otimes \psi_E$.
\begin{lemma} \label{charactperfectweak}
A state  $\ket{\psi_{a \alpha B E}}$ has perfect weak mutual independence if, and only if, there is an isometry $U : \alpha B \rightarrow \overline{a} \overline{E}$ such that 
\begin{equation} \label{perfectweakmutual}
U \ket{\psi_{a \alpha B E}} = \ket{\phi_{a \overline{a}}} \otimes \ket{\chi_{e \overline{e}}}
\end{equation}
\end{lemma}
\begin{proof}
It is clear that any state satisfying Eq. (\ref{perfectweakmutual}) has perfect weak mutual independence. To show the converse, we note that if $\psi_{aE} = \psi_a \otimes \psi_E$, there is a purification of $\psi_{aE}$ of the form $\ket{\phi_{a \overline{a}}} \otimes \ket{\chi_{e \overline{e}}}$. But as $\ket{\psi_{a \alpha B E}}$ is another purification of $\psi_{aE}$, by Ulhmann's theorem these two states must be related by an isometry acting on the purifying subsystem $\alpha B$. 
\end{proof}
\vspace{0.2 cm}

In \cite{Opp09} an intuitive explanation and generalised protocol
for Smith and Yard's 
examples of superactivation \cite{SY08} was presented. Given a pbit, one can use the erasure channel to send Alice's part of the \textit{shield} \cite{HHHO05} to Bob. This process then generates a state with coherent information equals half the size of the \textit{key} part of the pbit. Then, by considering the process in which Alice and Bob first distill pbits from their state and then use the erasure channel to convert them into distillable entanglement, we get Eq. (\ref{ssdistilableveruskey}).

A very similar protocol works also for mutual independence and weak mutual independence. For the latter, the $\alpha B$ register can be seen as the shield part of the state, which protects $\psi_a$ from having correlations with
$\psi_E$.  The amount of key is given by half the mutual information of the $a$ and $B$ registers. Now suppose Alice sends her share of 
the shield $\alpha$ through an erasure channel to Bob. With probability half, Bob gets $\alpha$. Then he can apply the isometry $U$ of Lemma \ref{charactperfectweak}, getting $S(a)$ of coherent information with Bob. In the case where Bob receives the erasure flag, they will end up with a state with coherent information $I(a \rangle B)$. Because Bob knows which case happened, the total coherent information is given by the average of the two, which is just the weak mutual independence of $\psi_{a \alpha B}$, $I(a : B)/2$. More formally, we have

\begin{theorem} \label{theoallthesame}
For a bipartite state $\psi_{AB}$,
\begin{equation}
W_{\text{ind}, E}(\psi_{AB}) = I_{\text{ind}, E}(\psi_{AB}) = D_{E}(\psi_{AB})
\end{equation}
and, likewise,
\begin{equation}
W_{\text{ind}, ss}(\psi_{AB}) = I_{\text{ind}, ss}(\psi_{AB}) = D_{ss}(\psi_{AB}).
\end{equation}
\end{theorem}

\begin{proof}
It is clear that $W_{\text{ind}, E}(\psi_{AB}) \geq I_{\text{ind}, E}(\psi_{AB}) \geq D_{E}(\psi_{AB})$. So let us show
\begin{equation} 
D_{E}(\psi_{AB}) \geq W_{\text{ind}, E}(\psi_{AB}).
\end{equation}
Consider the optimal protocol for $W_{\text{ind},E}(\psi_{AB})$. We can assume all the classical communication 
is made by using the erasure channel and that no system is discarded until the very end. In the final step 
the state shared by Alice, Bob and Eve is $\ket{\phi^n_{a \alpha : B : E}}$, with the property that 
\begin{equation} \label{asymptoticproductness1}
\lim_{n \rightarrow \infty} \Vert \phi^n_{a : E} - \phi^n_{a} \otimes \phi^n_{E} \Vert_1 =0, 
\end{equation}
and 
\begin{equation}
W_{\text{ind}, E}(\psi_{AB}) = \liminf_{n \rightarrow \infty} \frac{1}{2n} I(a : B)_{\phi^n_{a\alpha:B:E}}.
\end{equation}
Now, suppose that Alice instead of discarding $\alpha$, sends it down the erasure channel. Then the global state becomes
\begin{equation}
\frac{1}{\sqrt{2}} \ket{\phi^n_{a : B\alpha : E}} \otimes \ket{e}_{E'} + \frac{1}{\sqrt{2}} \ket{\phi^n_{a : B : E \alpha}} \otimes \ket{e}_{B'}.
\end{equation}
Due to the orthogonality of the erasure flag to the rest of the state, the coherent information of Alice's and Bob's subsystems splits into two as 
\begin{equation}
\frac{1}{2} I(a \rangle B) + \frac{1}{2} I(a \rangle B \alpha) = \frac{1}{2} I(a \rangle B) - \frac{1}{2} I(a \rangle E) 
\end{equation}
where we used the identity $I(X \rangle Y) = - I(X \rangle Z)$, valid for all pure states $\ket{\phi}_{XYZ}$. Then, by the hashing inequality \cite{DW05}, 
\begin{eqnarray}
D_{E}(\psi_{AB}) &\geq& \lim_{n \rightarrow \infty} \frac{1}{2} I(a \rangle B) - \frac{1}{2} I(a \rangle E)  \nonumber \\ &=& \lim_{n \rightarrow \infty} \frac{1}{2} I(a \rangle B) + \frac{1}{2} S(a) \nonumber \\ &=& \lim_{n \rightarrow \infty} \frac{1}{2} I(a : B) \nonumber \\ &=& W_{\text{ind}, E}(\psi_{AB}), \nonumber \\
\end{eqnarray}
where the before-last equality follows from the fact that $a$ and $E$ are asymptotically product (see Eq. \ref{asymptoticproductness1}). 

The proof for the ss-assisted quantities is completely analogous.
We only have to use the observation made in section \ref{backdef} that a symmetric-side channel can be used to simulate an erasure channel and apply the reasoning from before. 
\end{proof}
\vspace{0.2 cm}

\textbf{Remark 1:}
We note that we do not know whether the ss-assisted distillable secret-key is equal to the other quantities, and we conjectured that it is not. Indeed, for an EPR pair all the quantities are equal to one. For a pbit \cite{HHHO09}, in turn, the distillable secret-key is equal to one, while we only have the Smith-Yard bound $D_{ss}(\psi_{AB}) \geq K_{ss}(\psi_{AB})/2$ \cite{SY08}, which we expect to be tight when the shield part of the pbit is composed of a separable state. 

\textbf{Remark 2:} The theorem can be applied to any channel for which 
sending the shield part of a state with perfect weak mutual independence does not decrease the weak mutual independence rate (which is the case for the symmetric-side channel and the erasure channel): for any such channel $\Lambda$, we have $W_{\Lambda} =  D_{\Lambda}$.

\section{Superactivation and Weak Mutual Independence} \label{superactivation}

In this section we show that weak mutual independence is related to how much the erasure channel can activate the distillable entanglement of a given state, at least under a restricted class of protocols. Let $D^{(1)}_{E}(\psi_{AB})$ be the maximum 
coherent information assisted by an erasure channel, defined as
\begin{equation}
D^{(1)}_{E}(\psi_{AB}) := \max_{A \hookrightarrow a \alpha} I(a \rangle BB')_{\omega_{aBB'}}
\end{equation}
with $\omega_{a bB} := \Lambda_{E}(\psi_{a \alpha B})$ and $\Lambda_{E} : \alpha \hookrightarrow b$ an erasure channel, and where the maximization is taken over all isometric splittings of $A$ into $a \alpha$. 
We can fully characterize this quantity by weak mutual independence rate as follows.

\begin{proposition}
\begin{eqnarray}
D^{(1)}_{E}(\psi_{AB}) &=& W^{(1)}_{\text{ind}, \emptyset}(\psi_{AB}) \\ &=& \max_{A \rightarrow a \alpha} \frac{1}{2} \left( I(a:B) - I(a:E)\right). \nonumber
\end{eqnarray}
\end{proposition}

\begin{proof}
A simple calculation gives
\begin{equation}
D^{(1)}_{E}(\psi_{AB}) = \max_{A \hookrightarrow a \alpha} \frac{1}{2} I(a \rangle B \alpha)_{\psi_{a \alpha B}} + \frac{1}{2} I(a \rangle B)_{\psi_{a \alpha B}}.
\end{equation}

Writing the purification of $\psi_{a \alpha B}$ as $\psi_{a \alpha B E}$ we then find
\begin{eqnarray}
D^{(1)}_{E}(\Psi_{AB}) &=& \max_{A \hookrightarrow a \alpha}  \frac{1}{2} I(a \rangle B \alpha)_{\psi} + \frac{1}{2} I(a \rangle B)_{\psi} \nonumber \\ &=& \max_{A \hookrightarrow a \alpha} \frac{1}{2} I(a \rangle B \alpha)_{\psi} - \frac{1}{2} I(a \rangle E \alpha)_{\psi} \nonumber \\
&=&  \max_{A \hookrightarrow a \alpha} \frac{1}{2} I(a \rangle B)_{\psi} - \frac{1}{2} I(a \rangle E)_{\psi}. \nonumber \\
\end{eqnarray}
\end{proof}

\vspace{0.3 cm}
A corollary of the theorem is that there are states for which $W_{\text{ind}, \emptyset}(\psi_{AB}) \gg K_{\rightarrow}(\psi_{AB})$. An example is given by the Jamiolkowski state $\psi_{AB}$ of the \textit{rocket channel} of \cite{SS09}, where it was shown that $D_E^{(1)}(\psi_{AB}) \gg K_{\rightarrow}(\psi_{AB})$ (with the difference being of order of the number of qubits of Alice's state). 

By the same reasoning we also get that there are states for which the entanglement measure \textit{squashed entanglement} ($E_{\text{sq}}$) \cite{CW03} is much larger than the one-way distillable secret key. This comes from the observation
that the squashed entanglement is an upper bound on $W_{\text{ind}, \emptyset}$ (a fact proven implicitly in \cite{ADHW08} and explicitly in \cite{HOW09}) and so for the state associated with the rocket channel, 
\begin{equation}
E_{\text{sq}}(\psi_{AB}) \geq W_{\text{ind}, \emptyset}(\psi_{AB}) = D_E^{(1)}(\psi_{AB}) \gg K_{\rightarrow}(\psi_{AB}).
\end{equation}
Since the entanglement of formation $E_f$ can be much greater
than the squashed entanglement \cite{christandl-locking}, 
we also have that $E_f\gg W_{\text{ind}}$ is possible.

\section{Conclusions}

The capacity of a quantum channel is difficult to calculate because the best
formula we have~\cite{Llo97, Sho02, Dev05}, the coherent information, requires an optimisation
over an arbitrarily large number of usages of the channel.
The symmetric side-channel was originally introduced to  provide some insight into this optimisation --
it gives a more tractable upper-bound on the capacity of a channel, since when used in
conjunction with any channel, the combined capacity is single-letter.

However, here, and in \cite{BO10}, we have seen that the symmetric-side channel
is more than a calculational tool.  It should be thought of as playing the
role of public quantum communication, in the same way as public classical
communication makes the theory of private classical channels more elegant
and physically natural. The symmetric side-channel is conceptually
analogous to what one demands of a notion of 
public quantum communication (the receiver and eavesdropper both get the
same information), and it furthermore makes the rates for distilling
entanglement equivalent in form to the rates of distilling private key
using public classical communication.

The similarity between
entanglement and private correlations was used in
constructing the first entanglement distillation protocols, and
has been used
to conjecture new types of classical distributions \cite{GW00}. 
But the analogy between privacy and
entanglement was first made fully explicit by Collins and Popescu
\cite{CP02} and extended in \cite{OSW05}.  The identification of the symmetric side-channel
with public classical communication completes this analogy. 

We have seen that the introduction of public quantum communication
makes the rate for distilling entanglement have a similar form as the
rate for distilling classical key.  Likewise, it allows for correcting
noisy correlations, in much the same way as in the classical case.  This
gives further motivation for the study of mutual independence, and the
weak mutual independence introduced here.  It also helps us understand
the phenomena of superactivation, and more generally, non-additivity
of the channel capacity: the symmetric side-channel has no capacity,
but it helps correct errors introduced by a noisy quantum channel, in
much the same way that classical communication allows for error reconciliation
of private correlations.  A better understanding of weak mutual independence
thus provides a way to better understand superactivation and other forms
of non-additivity.

The mutual independence and the weak mutual independence also help us
understand tripartite correlations.  It quantifies how two parties
can be correlated (or decoupled) from a third party.  Mutual independence
quantifies how hard it is for the global state to be decoupled while
still retaining bipartite correlations, while the weak mutual independence
quantifies how hard it is for one party to decouple.  To quantify
how hard it is for each party to individually decouple from the third
party (while retaining bipartite correlations), one can 
consider the {\it not-so-weak mutual independence}, 
where we demand individual privacy
but not collective privacy: 
$\rho_{AE}\simeq\rho_A\otimes\rho_E$ and $\rho_{BE}\simeq\rho_B\otimes\rho_E$.
How this compares to the original mutual independence, enables one to quantify
the extent to which a third party is {\it correlated to the correlations}
of two parties. It thus enables a better understanding of genuine tripartite
correlations, a concept which is not well characterised even in the classical
case.  This is reminiscent of attempts to understand bipartite correlations
(and mutual information) in terms of the number of unitaries needed
to decouple a state from another \cite{GPW05}.

The work here raises a lot of open questions.  For example, we
do not even know if the erasure channel or the
 symmetric side-channel helps in distilling weak mutual independence.  For
that matter, it is possible that even a classical communication channel is not
helpful. It
is possible, that these additional channels are only useful for correcting errors (i.e. turning
mutual independence into EPR pairs).  It is also possible that the
erasure channel is as good as the symmetric-side channel for distilling
weak mutual-independence, and hence for superactivation of private states.  Indeed, for general
private states, the erasure channel appears to provide the optimal 
protocol~\cite{Opp09}. Using the results of Section \ref{backdef}
it would then be the optimal anti-degradable channel.

This work also suggests several questions about categories
of states and channels.
Are there states which have weak mutual independence, but from which
no private key can be distilled?  This is related to the problem of
{\it bound key} \cite{GW00}. 
Can we find a characterization of the convex set of states with zero $D_{ss}$?
Can we characterize the class of channels which have zero capacity for generating weak
mutual independence?
Can we use the connection between mutual independence and channel capacity
to find more examples of non-additivity and superactivation?

Finally, although Eq. (\ref{ssdistillableent}) is single-letter, it is of little practical use, as it involves an optimization over a system of unbounded dimension. Can we upper bound the size of the register that goes in the symmetric-side channel in the optimal protocol, in analogy to what can be done in the classical case?

\section*{Acknowledgements}
We are grateful to Graeme Smith for interesting discussions, 
and his helpful comments on an early draft of this work.  JO thanks 
the hospitality of KITP, Santa Barbara during his stay.

\appendix{Proof of Theorem \ref{ssanti}} \label{app1}

Let $U_{\Lambda} : A \hookrightarrow BE$ be an isometric extension of $\Lambda$. Since $\Lambda$ is anti-degradable, there is an isometry $V : E \hookrightarrow FG$ such that 
\begin{equation}
\ket{\phi}_{RBFG} := \left(V^{E \rightarrow FG} \circ U_{\Lambda}^{A\rightarrow BE}\right) \ket{\psi}_{RA} 
\end{equation}
have equal $RB$ and $RF$ reduced states (i.e. $\phi_{RB} = \phi_{RF}$) for every state $\ket{\psi_{RA}}$, where $R$ is a reference system. 

Define the state
\begin{eqnarray}
&&\ket{\Psi}_{RBFG_1G_2H_1H_2} := \frac{1}{\sqrt{2}}U_{ss, |G|}^{G \rightarrow G_1G_2}\ket{\phi}_{RBFG}\otimes\ket{00}_{H_1H_2} \nonumber \\ &+&\frac{1}{\sqrt{2}}\text{SWAP}_{BF}\left(U_{ss, |G|}^{G \rightarrow G_1G_2}\ket{\phi}_{RBFG}\right) \otimes \ket{11}_{H_1H_2}, \nonumber
\end{eqnarray}
with $\text{SWAP}_{BF}$ the unitary which swaps subsystems $B$ and $F$. By construction $\Psi_{RB} = \frac{1}{2}(\phi_{RB} + \phi_{RF}) = \phi_{RB}$ and so the states $\ket{\Psi}$ and $\ket{\phi}$ are related by an isometry from $FG$ to $FG_1G_2H_1H_2$. Therefore it suffices to show how Alice and Bob can create the state $\ket{\Psi}$  (with Eve holding the $FG_1G_2H_1H_2$ registers) by a symmetric-side channel and local operations. 

Note that $\ket{\Psi}$ is permutation symmetric with respect to the subsystems $\alpha := BG_1H_1$ and $\overline{\alpha} := FG_2H_2$. Thus there is a pre-image state $\ket{\Phi}_{RS}$ such that 
\begin{equation}
\ket{\Psi}_{RBFG_1G_2H_1H_2} = U_{ss, |S|}^{S\rightarrow \alpha \overline{\alpha}} \ket{\Phi}_{RS}.
\end{equation}

Alice and Bob's protocol is the following: Alice transforms $\ket{\psi_{RA}}$ into $\ket{\Phi}_{RS}$ by applying an operation ${\cal E} : A \rightarrow S$ (which always exists because $\psi_R = \Phi_R$). Then she sends $S$ through a symmetric-side channel producing the state $\ket{\Psi}$, where Alice and Eve hold the $BG_1H_1$ and $FG_2H_2$ registers, respectively. Finally, Bob applies his local operation ${\cal F}$ which consists of tracing out the register $G_1H_1$.


\begin{thebibliography}{99}

\bibitem{Mau93} U.M. Maurer. Secret Key Agreement by Public Discussion from Common Information. IEEE Trans. Info. Theory \textbf{39}, 733 (1993).

\bibitem{AC93} R. Ahlswede and I. Csiszar. Common Randomness in Information Theory and Cryptography Part II. IEEE Trans. Inf. Theory \textbf{39}, 1121 (1993).

\bibitem{CK78} I. Csiszar and J. Korner. Broadcast Channels with Confidential Messages.  IEEE Trans. Inf. Theory \textbf{24}, 339 (1978).
 
\bibitem{Wyn75} A. D. Wyner. The wire-tap channel. Bell Sys. Tech. J.\textbf{54}, 1355 (1975).

\bibitem{HHHH09} R. Horodecki, P. Horodecki, M. Horodecki, and K. Horodecki. Quantum Entanglement. Rev. Mod. Phys. Vol. 81, No. 2, pp. 865-942 (2009).
 
\bibitem{DW05} I. Devetak and A. Winter. Distillation of secret key and entanglement from quantum states. Proc. R. Soc. Lond. A \textbf{461}, 207 (2005). 

\bibitem{HHHO09} K. Horodecki, M. Horodecki, P. Horodecki, and J. Oppenheim. General paradigm for distilling classical key from quantum states. IEEE Trans. Inf. Theory \textbf{55}, 1898 (2009).

\bibitem{WZ82} W.K. Wootters and W.H. Zurek. A Single Quantum Cannot be Cloned. Nature \textbf{299}, 802 (1982).

\bibitem{Die82} D. Dieks. Communication by EPR devices. Physics Letters A \textbf{92}, 271 (1982).

\bibitem{SSW08} G. Smith, J.A. Smolin, A. Winter. The quantum capacity with symmetric side channels. IEEE Trans. Info. Theory \textbf{54}, 9, 4208 (2008). 

\bibitem{BO10} F.G.S.L. Brand\~ao and J. Oppenheim. The quantum one-time pad in the presence of an eavesdropper. arXiv:1004.3328.

\bibitem{BR03} P.O. Boykin and V. Roychowdhury. Optimal encryption of quantum bits. Phys. Rev. A 67, 042317 (2003). 

\bibitem{AMTW00} A. Ambainis, M. Mosca, A. Tapp, R. de Wolf. Private quantum channels. Proc. IEEE Conf. on Found. Comp. Sci. (FOCS), 2000.

\bibitem{Leu00} D.W. Leung. Quantum Vernam Cipher. quant-ph/0012077.

\bibitem{SW06} B. Schumacher and M.D. Westmoreland. Quantum mutual information and the one-time pad. Phys. Rev. A 74, 042305 (2006).

\bibitem{SY08} G. Smith and J. Yard. Quantum Communication With Zero-Capacity Channels. Science \textbf{321}, 1812 (2008).

\bibitem{HHHO05} K. Horodecki, M. Horodecki, P. Horodecki, and J. Oppenheim. Secure key from bound entanglement. Phys. Rev. Lett. \textbf{94}, 160502 (2005).

\bibitem{HPHH08} K. Horodecki, L. Pankowski, M. Horodecki, and P. Horodecki. Low dimensional bound entanglement with one-way distillable cryptographic key. IEEE Trans. Inf. Theory \textbf{54}, 2621 (2008). 
 
\bibitem{SS08} G. Smith and J.A. Smolin. Can non-private channels transmit quantum information? Phys. Rev. Lett. \textbf{102}, 010501 (2009).

\bibitem{LWZG09} K. Li, A. Winter, X. Zou, and G. Guo. The private capacity of quantum channels is not additive. Phys. Rev. Lett. \textbf{103}, 120501 (2009).

\bibitem{SS09} G. Smith and J.A. Smolin. Extensive nonadditivity of privacy. Phys. Rev. Lett. 103, 120503 (2009).

\bibitem{HOW09} M. Horodecki, J. Oppenheim, A. Winter. Quantum Mutual Independence. arXiv:0902.0912v2 [quant-ph]. 

\bibitem{DS05} I. Devetak and P.W. Shor. The capacity of a quantum channel for simultaneous transmission of classical and quantum information. Comm. Math. Phys. \textbf{256}, 287 (2005).

\bibitem{Sch95} B. Schumacher. Quantum Coding. Phys. Rev. A \textbf{51}, 2738 (1995).

\bibitem{ADHW08} A. Abeyesinghe, I. Devetak, P. Hayden, A. Winter. The Mother of all Protocols: Reconstructing Quantum Information's Family Tree. arXiv:quant-ph/0606255.

\bibitem{Opp09} J. Oppenheim. For quantum information, two wrongs can make a right. Science \textbf{321}, 1783 (2008).

\bibitem{CW03} M. Christandl and A. Winter. "Squashed Entanglemen'' - An Additive Entanglement Measure. J. Math. Phys. \textbf{45}, 829 (2004).

\bibitem{christandl-locking} M. Christandl and A. Winter. Uncertainty, Monogamy, and Locking of Quantum Correlations. IEEE Trans. Inf. Theory \textbf{51}, 3159 (2005). 

\bibitem{Llo97} S. Lloyd. Capacity of the noisy quantum channel. Phys. Rev. A \textbf{55}, 1613 (1997).

\bibitem{Sho02} P.W. Shor. The quantum capacity and coherent information. Unpublished lecture notes. Online at http://www.msri.org/publications/ln/msri/2002/quantumcrypto/shor/1; MSRI Quantum Information, Berkeley, 2002.

\bibitem{Dev05} I. Devetak. The private quantum capacity and Quantum Capacity of a Quantum Channel. IEEE Trans. Info. Theory \textbf{51}, 44 (2005).

\bibitem{GW00} N. Gisin and S. Wolf. Linking Classical and Quantum Key Agreement: Is There ''Bound Information''? Proceedings of CRYPTO 2000, Lecture Notes in Computer Science \textbf{1880}, 482 (2000).

\bibitem{CP02} D. Collins and S. Popescu. A classical analogue of entanglement. Phys. Rev. A 65, 032321 (2002).

\bibitem{OSW05} J. Oppenheim, R.W. Spekkens and A. Winter. A classical analogue of negative information. arXiv:quant-ph/0511247.

\bibitem{GPW05} B. Groisman, S. Popescu and A. Winter. On the quantum, classical and total amount of correlations in a quantum state. Phys. Rev. A \textbf{72}, 032317 (2005). 


\end{thebibliography}
\end{document}